\documentclass{ifacconf}

\usepackage{graphicx}      
\usepackage{natbib}        
\usepackage{amsmath,amssymb,amsfonts}
\usepackage{algorithmic}
\usepackage{graphicx}
\usepackage{textcomp}
\usepackage{epstopdf}
\usepackage{dsfont}
\usepackage{float}
\usepackage{subfigure}
\usepackage{bm}
\usepackage{easyReview}
\usepackage{mathtools}

\newcommand{\SO}{\operatorname{SO}}
\renewcommand{\so}{\mathfrak{so}}
\newcommand{\SE}{\operatorname{SE}}
\newcommand{\se}{\mathfrak{se}}

\DeclareMathOperator{\tr}{tr}
\DeclareMathOperator{\dom}{dom}  
\DeclareMathOperator{\Ad}{Ad}  
\DeclareMathOperator{\ad}{ad}

\DeclareMathOperator{\blkdiag}{blkdiag}
\DeclareMathOperator{\nulls}{Null}
\DeclareMathOperator{\rank}{rank}
\DeclareMathOperator{\kg}{kg}
\DeclareMathOperator{\m}{m}
\newcommand{\T}{^{\top}} 
\newcommand{\ie}{\textit{i.e.}}
\allowdisplaybreaks

\begin{document}
	\begin{frontmatter}
		
		\title{Distributed Hybrid Feedback for Global Pose Synchronization of Multiple Rigid Body Systems on $SE(3)$} 
		
		\thanks[footnoteinfo]{This work was supported in part by the National Natural Science Foundation of China under Grant 62403205, and in part by the National Sciences and Engineering Research Council of Canada, under Grant NSERC-DG RGPIN-2020-06270.}
		
		\author[First]{Fengyu Lin} 
		\author[First]{Miaomiao Wang} 
		\author[First]{Housheng Su}
		\author[Second]{Abdelhamid Tayebi}
		
		\address[First]{School of Artificial Intelligence and Automation, Huazhong University of Science and Technology, Wuhan 430074, China (e-mail: \{linfy2024, mmwang, shs\}@hust.edu.cn).}
		\address[Second]{Department of Electrical Engineering, Lakehead University, Thunder Bay, ON P7B 5E1, Canada, and with the Department of Electrical and Computer Engineering, Western University, London, ON N6A 3K7, Canada (e-mail: atayebi@lakeheadu.ca)}
		
		\begin{abstract}                
			This paper investigates the problem of pose synchronization for multiple rigid body systems evolving on the matrix Lie group $\SE(3)$. We propose a distributed hybrid feedback control scheme with global asymptotic stability guarantees using relative pose and group velocity measurements. The key idea consists of constructing a new potential function on $\SE(3) \times \mathbb{R}$ with a generalized non-diagonal weighting matrix, and a set of auxiliary scalar variables with continuous-discrete hybrid dynamics. Based on the new potential function and the auxiliary scalar variables, a geometric distributed hybrid feedback designed directly on $\SE(3)$ is proposed to achieve global pose synchronization. Numerical simulation results are presented to illustrate the performance of the proposed distributed hybrid control scheme.
		\end{abstract}
		
		
	\end{frontmatter}
	
	\section{Introduction}
	The design of decentralized control algorithms for pose (position and orientation) synchronization of multiple rigid body systems (MRBSs) is crucial in many practical applications, such as robotics \citep{bullo2009distributed}, satellites \citep{subbarao2008nonlinear}, unmanned aerial vehicles (UAVs) \citep{rashad2020fully} and spacecraft \citep{lawton2002synchronized}. Since the orientation/attitude of a rigid body evolves on the Special Orthogonal group $\SO(3)$, the motion of a rigid body in three-dimensional space inherently evolves on the Special Euclidean group $\SE(3)$. There are several typical approaches that address the coordination of MRBSs in three-dimensional space. The typical one considers the decoupled control framework on $\SO(3) \times \mathbb{R}^3$, whose key idea lies in separating the rigid body dynamics into independent rotational dynamics and translational dynamics, and then designing independent feedback laws for the attitude control and the position control separately, see for instance, \citep{schlanbusch2011synchronization,nair2007stable}. These decoupled approaches fundamentally disregard the intrinsic geometric coupling between the rotational and translational dynamics inherent to the Lie group $\SE(3)$, leading to non-rigorous stability results.  Other approaches relying on the dual quaternions to parameterize $\SE(3)$ have been proposed in \citet{filipe2013simultaneous,wang2012dual} and \citet{malladi2020rigid}, which allows to represent the pose in a compact form without requiring decoupling. 
	
	However, the unit-quaternion is known as a non-unique attitude representation (double cover of the space of rotations) and may cause the so-called unwinding phenomenon in practice \citep{Mayhew2011}. Few works consider the simultaneous coordination of attitude and position dynamics on the matrix Lie group $\SE(3)$, such as \citep{thunberg2016consensus, hatanaka2011passivity, igarashi2008passivity}. The method proposed in \citet{thunberg2016consensus},  leverages the Axis-Angle representation of the rotation to design the control laws. In \citet{hatanaka2011passivity} and \citet{igarashi2008passivity}, a passivity-based pose synchronization approach on $\SE(3)$ has been proposed. However, the synchronization control schemes developed in \citet{thunberg2016consensus, hatanaka2011passivity} and \citet{igarashi2008passivity} fail to guarantee global asymptotic stability (GAS) due to the topological properties of the motion space containing $\SO(3)$.  In fact, almost global asymptotic stability (AGAS)\footnote{AGAS refers to the fact that the equilibrium point is asymptotically stable and attractive from all initial conditions except a set of zero Lebesgue measure.} is the strongest stability result that one can achieve, with continuous time-invariant feedback, in spaces involving $\SO(3)$. 
	
	To overcome the above-mentioned topological obstruction, a hybrid feedback framework \citep{goebel2009hybrid} on the matrix Lie group $\SO(3)$, relying on a ``synergistic'' family of potential functions, has been introduced in \citet{mayhew2011synergistic}, leading to GAS guarantees. This hybrid feedback framework has been leveraged to solve the attitude tracking and estimation problems on $\SO(3)$ with GAS guarantees, see for instance, \citep{mayhew2013synergistic,lee2015global,berkane2016construction,wang2023nonlinear}. 
	However, most of these hybrid approaches rely on the compactness of the group $\SO(3)$, which makes it difficult to extend them to the global pose control problem on the noncompact group $\SE(3)$. Recently, a new hybrid feedback scheme relying on an auxiliary scalar variable with hybrid dynamics and a suitable potential function on $\SO(3)\times \mathbb{R}$ ($\SE(3)\times \mathbb{R}$) has been proposed in \citet{wang2021hybrid}, which can achieve global attitude (pose) control on matrix Lie groups $\SO(3)$ ($\SE(3)$). Inspired by this work, the authors in \citet{boughellaba2024global} proposed a distributed attitude synchronization control scheme on $\SO(3)$ with GAS guarantees. However, the design of pose synchronization schemes directly on $\SE(3)$ with GAS guarantees remains an open problem. 
	
	In this work, we consider the problem of pose synchronization for MRBSs on $\SE(3)$ using relative pose information. We first provide a systematic procedure for the construction of a generic potential function on $\SE(3) \times \mathbb{R}$ that relaxes the requirement of a diagonal weighting matrix used in \citet{wang2021hybrid}. Then, an extended synergistic potential function on $(\SE(3)\times \mathbb{R})^N$ is introduced for the pose synchronization problem. With this synergistic potential function, a new distributed hybrid feedback for pose synchronization of MRBSs is designed directly on the group $\SE(3)$ with GAS guarantees.

	\section{Preliminaries}
	\subsection{Notations}
	The set of real, positive real, nonnegative real and the natural numbers are denoted by $\mathbb{R}, \mathbb{R}_{> 0},  \mathbb{R}_{\geq 0}$, and $\mathbb{N}$, respectively. 
	We denote by $\mathbb{R}^n$ the $n$-dimensional Euclidean space, and denote by $\mathbb{S}^n:= \{x \in \mathbb{R}^{n+1}|x^{\top} x=1\}$ the set of unit vectors in $\mathbb{R}^{n+1}$. For matrices $A, B \in \mathbb{R}^{m\times n}$, the Euclidean inner product is defined as $\langle\langle A, B \rangle\rangle=\tr(A^{\top} B)$. The Euclidean norm of a vector $x\in \mathbb{R}^n$ is defined as $\|x\|=\sqrt{x^{\top} x}$. The matrix $I_n\in \mathbb{R}^{n\times n}$ denotes the identity matrix. Denote the block diagonal matrix by $\blkdiag(\cdot)$. For a given matrix $A \in \mathbb{R}^{n\times n}$, we define the set of all unit-eigenvectors and the $i$-th pair of eigenvalue and eigenvector of $A$ by $\mathcal{E}_v(\cdot)$ and $(\lambda_i^A, v_i^A)$, respectively. Let $\mathbf{0}$ denote a vector or matrix with compatible dimensions. Let $\mathcal{D}$ be a smooth manifold embedded in $\mathbb{R}^n$ with $T_x\mathcal{D}$ being its tangent space at point $x \in \mathcal{D}$. Let $f: \mathcal{D} \rightarrow \mathbb{R}_{\geq 0}$ be a continuously differentiable real-valued function with respect to the set $\mathcal{A} \subset \mathcal{D}$ if $f(x) >0$ for all $x\notin \mathcal{A}$ and $f(x)=0$ for all $x \in \mathcal{A}$. Then, the function $f$ is a potential function on $\mathcal{D}$ with respect to the set $\mathcal{A}$. The gradient of $f$ at point $x\in \mathcal{D}$, denoted by $\nabla_xf(x) \in T_x \mathcal{D}$, is uniquely defined by $\dot{f}(x)=\langle \nabla_{x}f(x), \xi\rangle_x$ for all $\xi \in T_x\mathcal{D}$, where $\langle , \rangle_x: T_x\mathcal{D} \times T_x \mathcal{D} \rightarrow \mathbb{R}$ denotes the Riemannian metric on $\mathcal{D}$.  The point $x\in \mathcal{D}$ is called a critical point of $f$ if $\nabla_xf(x)=0$. 
	The cross product between two vectors $y, z\in \mathbb{R}^3$, can be written as the matrix multiplication $y \times z=y^{\times}z$ with 
	\begin{align*}
		y^{\times}=\begin{bmatrix}
			0&-y_3&y_2\\
			y_3&0&-y_1\\
			-y_2&y_1&0
		\end{bmatrix}.
	\end{align*}
	Given a matrix $A = [a_{ij}]_{1\leq i,j \leq 3}\in \mathbb{R}^{3\times 3}$, 
	define the map $\psi(A):= \frac{1}{2}[a_{32}-a_{23}, a_{13}-a_{31}, a_{21}-a_{12}]^{\top}$, such that $\langle \langle A,y^{\times} \rangle \rangle=2y^{\top} \psi(A)$ for all $A\in \mathbb{R}^{3\times 3}, y \in \mathbb{R}^3$.  Let the map $\mathcal{R}_a: \mathbb{R}\times \mathbb{S}^2 \rightarrow \SO(3)$ denote the well-known angle-axis parameterization of the attitude, given by $\mathcal{R}_a(\theta,u):= I_3+\sin(\theta)u^{\times}+(1-\cos(\theta))(u^{\times})^2$ with $\theta \in \mathbb{R}$ being the rotation angle and $u \in \mathbb{S}^2$ being the rotation axis.

	\subsection{Special Euclidean Group $\SE(3)$}
	The attitude $R$ of a rigid body is represented by a 3-by-3 rotation matrix that belongs to the three-dimensional Special Orthogonal group given by $ \SO(3):=\{R\in \mathbb{R}^{3\times 3}:R^{\top} R=R R^{\top}=I_3, \det(R)=1\}$. 
	The Lie algebra of $\SO(3)$ is denoted by $\so(3):=\{\Omega \in \mathbb{R}^{3 \times 3}:~ \Omega^{T}=-\Omega\}$.
	The pose $(R, p)$ of a rigid body is represented by a 4-by-4 matrix  $X= \mathcal{T}(R,p)$ that belongs to the Special Euclidean group:
	\begin{align*}
		\SE(3):=\Big\{X=\begin{bmatrix}
			R&p\\
			\mathbf{0}&1
		\end{bmatrix} \in \mathbb{R}^{4\times 4}: R\in \SO(3), p\in \mathbb{R}^3 \Big\}.
	\end{align*}
	The Lie algebra of $\SE(3)$ is denoted by
	\begin{align*}
		\se(3)=\Big\{U=\begin{bmatrix}
			\omega^{\times}&v\\
			\mathbf{0}&0
		\end{bmatrix} \in \mathbb{R}^{4\times 4}:\omega^{\times} \in \so(3), v \in \mathbb{R}^3\Big\}.
	\end{align*}
	We introduce the following map $(\cdot)^{\wedge}: \mathbb{R}^6 \rightarrow \se(3)$ as 
	\begin{align}\label{xi_def}
		\xi^{\wedge}=\begin{bmatrix}
			\omega^{\times}&v\\
			\mathbf{0}&0
		\end{bmatrix} \in \se(3)
	\end{align}
	with $\xi=[\omega^{\top},v^{\top}]^{\top}$ and $\omega, v\in \mathbb{R}^3$. The map  $\bar{\psi}:\mathbb{R}^{4\times 4} \rightarrow \mathbb{R}^6$ is defined as:
	\begin{align}\label{psiA_def}
		\bar{\psi}(\mathbb{A})=\begin{bmatrix}
			\psi(A)\\
			\frac{1}{2}b
		\end{bmatrix}, \quad \forall \mathbb{A}=\begin{bmatrix}
			A&b\\
			c^{\top}&d
		\end{bmatrix} \in \mathbb{R}^{4\times 4}
	\end{align}
	with $A\in \mathbb{R}^{3\times 3},b, c \in \mathbb{R}^3, d \in \mathbb{R}$. Then, one has the following identities:
	\begin{subequations}
		\begin{align}
			\langle\langle \mathbb{A}, x^{\wedge} \rangle\rangle&=2x^{\top}\bar{\psi}(\mathbb{A}) \label{identa} \\ 
			\bar{\psi}(X^{\top}(I_4-X)\mathbb{A})&=-\bar{\psi}((I_4-X^{-1})\mathbb{A})  \label{identb}
		\end{align}
	\end{subequations}
	for all $\mathbb{A} \in \mathbb{R}^{4\times 4}, x \in \mathbb{R}^6, X\in \SE(3)$. Define the adjoint map $\Ad: \SE(3) \rightarrow \mathbb{R}^{6\times 6}$ as 
	\begin{align}\label{Ad_def}
		\Ad_{X}=\begin{bmatrix}
			R&\mathbf{0}\\
			p^{\times}R&R
		\end{bmatrix} \in \mathbb{R}^{6\times 6}, \forall X=\begin{bmatrix}
			R&p\\
			\mathbf{0}&1
		\end{bmatrix}
	\end{align}
	and the adjoint operator $\ad: \mathbb{R}^6 \rightarrow \mathbb{R}^{6\times 6}$ as 
	\begin{align}\label{ad_def}
		\ad_{\xi}=\begin{bmatrix}
			\omega^{\times}& \mathbf{0}\\
			v^{\times}& \omega^{\times} 
		\end{bmatrix} \in \mathbb{R}^{6\times 6}, \quad \forall \xi= \begin{bmatrix}
			\omega\\v
		\end{bmatrix}. 
	\end{align}
	Then, one has the following useful identities:
	\begin{subequations}\label{props}
		\begin{align}
			\Ad_{X^{-1}}&=\Ad_{X}^{-1} \label{propa}\\
			Xx^{\wedge}X^{-1}&=(\Ad_{X}x)^{\wedge} \label{propb} \\
			\ad_x x&=\mathbf{0} \label{propc} 
		\end{align}
	\end{subequations}
	for all $X \in \SE(3), x\in \mathbb{R}^6$.

	\subsection{Graph Theory}
	Consider a network comprising $N$ rigid body systems, where the information flow between the individual systems is modeled by a graph $\mathcal{G}=(\mathcal{V}, \mathcal{E})$, with $\mathcal{V}=\{1,\dots, N\}$ and $\mathcal{E} \subseteq \mathcal{V} \times \mathcal{V}$ denoting the vertex and edge sets of graph $\mathcal{G}$, respectively. In an undirected graph, an edge $(i,j)\in \mathcal{E}$ indicates that the $i$-th and $j$-th rigid bodies interact with each other, which means information can be exchanged bidirectionally between them. When two rigid-body systems can access each other’s relative information, they are connected by an edge in the graph. For each edge between two rigid bodies, one arbitrarily assigns an index, a positive end, and a negative end. Let $\mathcal{M}=\{1,\dots, M\}$ denote the set of graph edges with $M$ denoting the total number of edges, and let $\mathcal{M}_i^+\subset\mathcal{M}$ and $\mathcal{M}_i^-\subset\mathcal{M}$ denote the set of edges for which node $i$ is the positive end and the negative end, respectively. Define $\mathcal{N}_i:=\{j\in \mathcal{V}: (i,j)\in \mathcal{E}\}$ as the set of connected neighbors of the $i$-th rigid body, and $\mathcal{M}_i :=\mathcal{M}_i^+ \cup \mathcal{M}_i^-$ denote the set of edges connected to the $i$-th rigid body. The classical incidence matrix $B: =[B_{ik}]_{\substack{1\leq i \leq N, 1\leq k \leq M}} \in \mathbb{R}^{N\times M}$ is given by
	\begin{align}\label{B_def}
		B_{ik}=\left\{
		\begin{array}{ll}
			+1 \quad  & k \in \mathcal{M}_i^+ \\
			-1 \quad  & k \in \mathcal{M}_i^-\\
			0 \quad & otherwise 
		\end{array} ,~~~ \forall i \in \mathcal{V},  k\in \mathcal{M}.
		\right.
	\end{align}
	\section{Problem Statement}
	Consider an MRBS with a set of $N$ rigid bodies. Let $X_i=\mathcal{T}(R_i,p_i)\in \SE(3)$ denote the pose of the $i$-th rigid body system, and $\xi_i=[\omega_i^{\top}, v_i^{\top}]^{\top} \in \mathbb{R}^6$ denote the group velocity of the $i$-th rigid body system  expressed in the body-fixed frame. For each rigid body system, the fully actuated dynamics on $\SE(3)\times\mathbb{R}^6$ are given by \citep{bullo1999tracking}
	\begin{align}\label{eqn:kinematic}
		\begin{cases}
			\dot{X}_i&=X_i\xi_i^{\wedge} \\
			\mathbb{I}_i\dot{\xi}_i&=\ad_{\xi_{i}}^{\top}\mathbb{I}_i\xi_i+u_{i}
		\end{cases}
	\end{align} 
	where $\mathbb{I}_i=\blkdiag(J_i,m_iI_3)  \in \mathbb{R}^{6\times 6}$ denotes the mass-inertia matrix of the $i$-th rigid body with $J_i=J_i^{\top}>0$ and $m_i>0$, and $u_{i}=[\tau_i^{\top}, f_i^{\top}]^{\top} \in \mathbb{R}^6$ with $\tau_i$ and $f_i$ denoting the torque and force inputs of the $i$-th rigid body, respectively.  An example of a fully actuated multirotor UAV can be found in  \cite{rashad2020fully}.
	The interaction topology among rigid bodies is described by the graph $\mathcal{G}$, which implies that the relative pose is available between two rigid bodies if they are connected, \ie, $(i,j)\in \mathcal{E}$.
	By assigning an arbitrary orientation to each edge $k\in  \mathcal{M}_i^+ \cap \mathcal{M}_j^- \subset \mathcal{M}$ connecting rigid bodies $(i,j)\in \mathcal{E}$, we define the relative pose between them as $\bar{X}_k=\mathcal{T}(\bar{R}_k,\bar{p}_k):=X_j^{-1}X_i$.
	
	Then, in view of \eqref{eqn:kinematic}, one has
	\begin{align}    \label{eqn:errdynamics}  
		\begin{cases}
			\dot{\bar{X}}_k&=\bar{X}_k\bar{\xi}_k^{\wedge} \\                   \mathbb{I}_i\dot{\xi}_i&=\ad_{\xi_{i}}^{\top}\mathbb{I}_i\xi_i+u_{i} 
		\end{cases}   
	\end{align} 
	where $\bar{\xi}_k:=\xi_i-\Ad_{\bar{X}_k}^{-1}\xi_j$. It is obvious that for every $i \in \mathcal{V}$ and $j \in \mathcal{N}_i$, the intersection between the sets $\mathcal{M}_i^+$ and $\mathcal{M}_j^-$ is either a set with a single element or an empty set otherwise. Let $\bar{\xi}=(\bar{\xi}_1^{\top}, \bar{\xi}_2^{\top},\dots, \bar{\xi}_M^{\top})^{\top}
	\in \mathbb{R}^{6M}$ and $\xi=(\xi_1^{\top}, \xi_2^{\top}, \dots, \xi_N^{\top})^{\top}
	\in \mathbb{R}^{6N}$. It can be verified that
	\begin{align}\label{eqn:bar_xi}
		\bar{\xi}=\bar{B}(t)^{\top}\xi
	\end{align}
	where the incidence matrix $\bar{B}(t):=[\bar{B}_{ik}]_{\substack{1\leq i \leq N, 1\leq k \leq M}} \in \mathbb{R}^{6N\times 6M}$ is defined as follows:
	\begin{align}\label{Bt_def}
		\bar{B}_{ik}=\left\{
		\begin{array}{ll} I_6&k\in \mathcal{M}_i^+\\
			-\Ad_{\bar{X}_k}^{-\top} & k \in \mathcal{M}_i^{-}\\
			\mathbf{0}& otherwise \end{array}\right. , \forall i \in \mathcal{V},  k\in \mathcal{M}.
	\end{align}
	
	Unlike the classical incidence matrix $B$ defined in \eqref{B_def}, the matrix $\bar{B}(t)$ defined in \eqref{Bt_def} is time-varying and depends on the relative pose $\bar{X}_k$ of each edge. One introduces the following important assumption for the graph $\mathcal{G}$:
	\begin{assum} \label{assum:1}
		The interaction graph $\mathcal{G}$ for system \eqref{eqn:kinematic} is assumed to be connected and acyclic.
	\end{assum}
	
	Note that Assumption \ref{assum:1} is commonly considered in the problem of attitude/pose synchronization \citep{wang2012dual,boughellaba2024global}. In the following lemma, we present a useful property of $\bar{B}(t)$ under Assumption \ref{assum:1}:      
	\begin{lem}\label{lem: HE} 
		Consider the matrix $\bar{B}(t)$ defined in \eqref{Bt_def} for system  under Assumption \ref{assum:1}. Then, $\bar{B}(t)x=\mathbf{0}$ implies $x=\mathbf{0}$ for all $t \geq 0$.
	\end{lem}
	\begin{pf}
		Inspired by \cite{boughellaba2024global}, we will prove this lemma by contradiction.  Note that the graph is connected and acyclic, which implies that the graph contains a tree and the number of edges is $M=N-1$.  
		Then, one has rank$(\bar{B}(t)) \leq 6N-6$ due to $\bar{B}(t) \in \mathbb{R}^{6N \times 6M}$. Moreover, let $y=[y_1^{\top},\dots,y_N^{\top}]^{\top} \in \mathbb{R}^{6N}$ be a vector that belongs to the null space of $\bar{B}^{\top}(t), \ie,  \bar{B}^{\top}(t)y=\mathbf{0}$. According to \eqref{Bt_def}, one has $y_j=Ad_{\bar{X}_k}y_i$ owing to the invertibility of $Ad_{\bar{X}_k}$, for all $(i,j)\in \mathcal{E}$ and $k\in \mathcal{M}$. Due to the connectivity of the graph $\mathcal{G}$, without loss of generality, there exists a map $\mathcal{Q}_j: (\mathbb{R}^{6\times 6})^{M} \rightarrow \mathbb{R}^{6 \times 6}$ such that $y_j=\mathcal{Q}_j (Ad_{\bar{X}_1}, \dots, Ad_{\bar{X}_M})y_1$, for each $j\in \mathcal{V}$. Define $\mathcal{Q}: =[\mathcal{Q}_1^{\top}, \dots, \mathcal{Q}_N^{\top}]^{\top} \in \mathbb{R}^{6N \times 6}$, one can show that $y=\mathcal{Q}(Ad_{\bar{X}_1}, \dots, Ad_{\bar{X}_M})y_1$. 
		The injectivity of $\mathcal{Q}$, which follows from the invertibility of each $\mathcal{Q}_j$, implies that the solution space of $\bar{B}(t)^{\top} y = \mathbf{0}$ is isomorphic to $y_1 \in \mathbb{R}^6$.  Therefore, it is clear $\dim \nulls (\bar{B}^{\top}(t))=6$. Together with rank-nullity theory, one can obtain $\rank(\bar{B}(t))=\rank(\bar{B}(t)^{\top})=6N-6$, which means that the matrix $\bar{B}(t)$ is full of column rank. It follows that $x=\mathbf{0}$ is the only solution to $\bar{B}(t)x=\mathbf{0}$.  This completes the proof.
	\end{pf}

	The objective of this work is to design a distributed control scheme for the MRBSs described by \eqref{eqn:kinematic} under Assumption \ref{assum:1} such that the compact set $\mathcal{O}:=\{\bar{X}_k = I_4, \xi_i=\mathbf{0}, \forall k\in \mathcal{M}, i\in \mathcal{V}\}$ is globally asymptotically stable.

	\section{Potential function design on $(\SE(3) \times \mathbb{R})^M$} 
	The construction of the potential function is critical for achieving pose synchronization of MRBSs. In this section, a systematic procedure for constructing a general potential function on $(\SE(3) \times \mathbb{R})^M$ is provided.
	\subsection{Potential Function on $\SE(3)$ and $\SE(3)\times \mathbb{R}$}\label{IVA}
	Consider the following real-valued function $V:\SE(3) \to \mathbb{R}$:
	\begin{align}\label{V_PF}
		V(\bar{X}_k):=\frac{1}{2}\tr((I_4-\bar{X}_k)\mathbb{A}(I_4-\bar{X}_k)^{\top})
	\end{align}
	where the weighting matrix $\mathbb{A}=\mathbb{A}\T$ is given by
	\begin{align}\label{AA_def}
		\mathbb{A}:=\begin{bmatrix}
			A&b\\
			b^{\top}&d
		\end{bmatrix} \in \mathbb{R}^{4\times4}
	\end{align}
	with $A \in \mathbb{R}^{3\times 3}, b\in \mathbb{R}^3, d\in \mathbb{R}_{>0}$, $W:= A-bb^{\top}d^{-1}$ being positive semi-definite, and $\bar{W}:=\frac{1}{2}(\tr(W)I_3-W)$ being positive definite. It follows from the definition of $\bar{W}$ that $\mathcal{E}_v(\bar{W}) \equiv \mathcal{E}_v(W)$. Then, one verifies that $V$ is a potential function on $\SE(3)$ with respect to $I_4$, \ie, $V(\bar{X}_k)\geq 0$ for all $\bar{X}_k\in \SE(3)$ and $V(\bar{X}_k)=0$ if and only if $\bar{X}_k = I_4$. 
	Applying the following matrix decomposition:
	\begin{align}
		\begin{bmatrix}
			A&b\\
			b^{\top}&d
		\end{bmatrix}=\begin{bmatrix}
			I_3&bd^{-1} \nonumber \\
			\mathbf{0}&1
		\end{bmatrix}
		\begin{bmatrix}
			W&0\\
			\mathbf{0}&d
		\end{bmatrix}
		\begin{bmatrix}
			I_3&\mathbf{0}\\
			b^{\top}d^{-1}&1
		\end{bmatrix}
	\end{align}
	one can easily rewrite the potential function $V$ in \eqref{V_PF} as
	\begin{align}\label{V_PF_explict}
		V(\bar{X}_k)
		&=\frac{1}{2}\tr((I_3-\bar{R}_k)W(I_3-\bar{R}_k)^{\top}) \nonumber \\
		&~~~~+\frac{d}{2}\|\bar{p}_k-(I_3-\bar{R}_k)bd^{-1}\|^2.
	\end{align} 
	The first term of $V$ in \eqref{V_PF_explict} is the well-known potential function on $\SO(3)$ based on a modified trace function, while the second term of $V$ in \eqref{V_PF_explict} contains both relative attitude and position when $b\neq \mathbf{0}$. 
	Let $\nabla_{\bar{X}_k}V$ denote the gradient of $V$ at point $\bar{X}_k$.
	From \citep{wang2018hybrid}, the set of all the critical points of $V(\bar{X}_k)$ is given by 
	\begin{align}\label{def:Upsilon_V}
		\Upsilon_V :=&\{\bar{X}_k \in \SE(3): \nabla_{\bar{X}_k}V(\bar{X}_k)=0\} \nonumber \\
		=& \{I_4\} \cup \{\bar{X}_k=\mathcal{T}(\bar{R}_k, \bar{p}_k): \bar{R}_k=\mathcal{R}_a(\pi, v),  \nonumber \\
		&\qquad  \quad \bar{p}_k=(I_3-R_a(\pi, v))bd^{-1}, v\in \mathcal{E}_v(W)\}
	\end{align} 
	and the set of all the undesired critical points is denoted by $\Upsilon_V \setminus \{I_4\}$. 
	As shown in \cite{koditschek1989application}, no smooth potential function with a global attractor exists on Lie groups due to topological obstructions. Therefore, gradient-based controllers, relying on such a potential function, can at most achieve AGAS.

	In order to achieve GAS, a modified potential function on $\SE(3) \times \mathbb{R}$ is proposed in \cite{wang2021hybrid}. Consider the transformation map $\Gamma_{\mathbb{A}}: \SE(3) \times \mathbb{R} \rightarrow \SE(3)$ defined as 
	\begin{align}\label{Gamma_def}
		\Gamma_{\mathbb{A}}(\bar{X}_k, \theta_k)
		:=\bar{X}_k \exp(\theta_k u_c^{\wedge})
	\end{align}
	where a constant vector $u_c\in \mathbb{R}^6$ and a real-valued scalar $\theta_k \in \mathbb{R}$ with hybrid dynamics. Then, a modified potential function $U: \SE(3) \times \mathbb{R}  \rightarrow \mathbb{R}_{\geq 0}$ with respect to $\mathcal{A}_0:=(I_4,0)$ is given as:
	\begin{align}\label{U_def}  
		U(\bar{X}_k, \theta_k):=\frac{1}{2}\tr((I_4-\Gamma_{\mathbb{A}})\mathbb{A} (I_4-\Gamma_{\mathbb{A}})^{\top})+\frac{\gamma}{2}\theta_k^2
	\end{align}
	with some $\gamma>0$. It is easy to verify from \eqref{U_def} that $U(\bar{X}_k,\theta_k)> 0$ for all $(\bar{X}_k, \theta_k) \in \SE(3) \times \mathbb{R} $, and $U(\bar{X}_k,\theta_k)= 0$ if and only if $(\bar{X}_k, \theta_k)=\mathcal{A}_0$. Therefore, $U(\bar{X}_k,\theta_k)$ in \eqref{U_def} is a potential function on $\SE(3) \times \mathbb{R}$ with respect to $\mathcal{A}_0$. In \cite{wang2021hybrid}, by considering a block diagonal matrix $\mathbb{A}= \blkdiag(A,d)$ and $u_c = [u_{c1}\T, \mathbf{0}\T]\T$ with $u_{c1} \in \mathbb{S}^2$, the potential function in \eqref{U_def} can be simplified as:
	\begin{align*}
		U(\bar{X}_k, \theta_k)=\tr(W(I_3-\bar{R}_k\mathcal{R}_a(\theta_k, u_{c1})))+\frac{\gamma}{2}\theta_k^2+\frac{d}{2}\|\bar{p}_k\|^2.
	\end{align*}
	It is clear that the third term in $U(\bar{X}_k, \theta_k)$ is independent of the rotation. This can simplify the hybrid feedback control design, resulting in a strict decrease in the potential function over each jump as shown in \cite{wang2021hybrid}. 
	However, the choice of a block diagonal matrix $\mathbb{A}$ is very restrictive in practical applications. 
	For instance, in some applications such as in \cite{wang2018hybrid}, the matrix $\mathbb{A}$ is directly constructed from landmark positions $p_i \in \mathbb{R}^3$, where $\mathbb{A}:= \sum_{i=1}^{n}k_ir_ir_i^{\top}$ with $r_i=\begin{bmatrix}
		p_i^{\top}&1
	\end{bmatrix}^{\top} \in \mathbb{R}^4$ and $k_i > 0$. In this situation, it is difficult to obtain the matrix $\mathbb{A}$ in diagonal form, due to the  physical location of the landmarks.

	\subsection{Generic Potential Function on $(\SE(3)\times \mathbb{R})^M$}
	To overcome the above-mentioned limitation, we construct a generic potential function on $\SE(3) \times \mathbb{R}$ for a general matrix $\mathbb{A}$ defined in \eqref{AA_def}. Consider $\theta_k \in \mathbb{R}$ and $u_c=[u_{c1}^{\top}, u_{c2}^{\top}]^{\top} \in \mathbb{R}^6$ with $u_{c1} \in \mathbb{S}^2$ and $u_{c2} \in \mathbb{R}^3$. Then, one can show that 
	\begin{align}\label{eqn:exp_thetak}
		\exp(\theta_k u_c^{\wedge}) &=\exp\begin{pmatrix}
			\theta_k u_{c1}^{\times}&  ~~~  \theta_k u_{c2} \nonumber \\
			\mathbf{0}&0
		\end{pmatrix}\\
		&= \begin{pmatrix}
			\exp(\theta_k u_{c1}^\times) & ~~~ \mathcal{U}(\theta_k u_{c1}^{\times}) \theta_k u_{c2}\\
			\mathbf{0}&1
		\end{pmatrix}
	\end{align}
	where the map $\mathcal{U}$ is defined as
	$ 
	\mathcal{U}(\theta_k u_{c1}^{\times}):=
	I_3+(\frac{1-\cos (\theta_k)}{\theta_k})u_{c1}^{\times}+(\frac{\theta_k-\sin (\theta_k)}{\theta_k})(u_{c1}^{\times})^2
	$. Then, using the facts $\mathcal{R}_a(\theta_k,u_{c1}) = I_3+\sin (\theta_k) u_{c1}^{\times}+(1-\cos(\theta_k))(u_{c1}^{\times})^2$ and $(u_{c1}^{\times})^3 = -\|u_{c1}\|^2 u_{c1}^{\times} = -u_{c1}^{\times}$, one can verify that
	\begin{align}\label{mathcalU}
		\mathcal{U}(\theta_k u_{c1}^{\times})\theta_k u_{c1}^\times  = \mathcal{R}_a(\theta_k ,u_{c1})-I_3. 
	\end{align}
	For each $\mathbb{A}$ defined in \eqref{AA_def}, we pick constant $u_c:=[u_{c1}^{\top},u_{c2}^{\top}]^{\top} \in \mathbb{R}^6$ for the design of the transformation map $\Gamma_{\mathbb{A}}$ in \eqref{Gamma_def} with $u_{c1} \in \mathbb{S}^2$ and $u_{c2}=-u_{c1}^{\times} bd^{{-1}}$, where  $d\in\mathbb{R}_{>0}, b\in \mathbb{R}^3$ are given by \eqref{AA_def} and the unit vector $u_{c1}$ will be given later.
	Then, from \eqref{eqn:exp_thetak} and \eqref{mathcalU} one has 
	\begin{align*}
		\mathcal{T}_{u_c}(\theta_k)&:=\exp(\theta_k u_c^{\wedge}) \\ 
		&=\begin{pmatrix} 
			\mathcal{R}_a(\theta_k,u_{c1})& ~~~ (I_3-\mathcal{R}_a(\theta_k,u_{c1}))bd^{-1}\\
			\mathbf{0}&1
		\end{pmatrix}.
	\end{align*} 
	Hence, from \eqref{Gamma_def} and \eqref{U_def} one has the following explicit potential function $U(\bar{X}_k, \theta_k)$:
	\begin{align}\label{U_explicit}
		U(\bar{X}_k, \theta_k)
		&=\tr(W(I_3-\bar{R}_k\mathcal{R}_a(\theta_k, u_{c1})))+\frac{\gamma}{2}\theta_k^2 \nonumber\\
		&~~~~+\frac{d}{2}\|\bar{p}_k - (I_3 - \bar{R}_k)bd^{-1}\|^2. 
	\end{align}
	It follows from \eqref{U_explicit} that a change in $\theta_k$ will only affect the first two items $\tr(W(I_3-\bar{R}_k\mathcal{R}_a(\theta_k, u_{c1})))+\frac{\gamma}{2}\theta_k^2$, while leaving the third item $\frac{d}{2}\|\bar{p}_k - (I_3 - \bar{R}_k)bd^{-1}\|^2$ unchanged. This nice feature will simplify the design of our hybrid feedback in the next section to avoid all the undesired equilibrium points through jumps and achieve GAS for the hybrid closed-loop system. Note that a similar result was achieved in \cite{wang2021hybrid} only for a block diagonal matrix $\mathbb{A} = \blkdiag(A,d)$ and $u_{c2} = \mathbf{0}$, which is a special case (\ie, $b= \mathbf{0}$ and $u_{c2} = -u_{c1}^\times bd^{-1} = \mathbf{0}$) of the potential function considered in this work.
	
	Define the set of all the critical points of $U$ in \eqref{U_def} as 
	\begin{align}\label{def_Upsilon_U}
		\Upsilon_U:=\{&(\bar{X}_k, \theta_k) \in \SE(3) \times \mathbb{R}: \nonumber \\
		&\nabla_{\bar{X}_k}U(\bar{X}_k, \theta_k) = \mathbf{0}, \nabla_{\theta_k}U(\bar{X}_k, \theta_k)=0\}.   
	\end{align}
	The following lemma  provides some useful properties of the map $\Gamma_{\mathbb{A}}$ in \eqref{Gamma_def} and the potential function $U$ in \eqref{U_def}:
	\begin{lem}{\cite[Proposition 3]{wang2021hybrid}}\label{lem:Uproperty} 
		Consider the transformation map $\Gamma_{\mathbb{A}}$ defined in \eqref{Gamma_def}, and the trajectories generated by $\dot{\bar{X}}_k=\bar{X}_k\bar{\xi}_k^{\wedge}$ and $\dot{\theta}_k=\nu_k$ with $(\bar{X}_k(0),\theta_k(0) )\in \SE(3)\times \mathbb{R}, (\bar{\xi}_k, \nu_k) \in  \mathbb{R}^6 \times \mathbb{R}$. Then, the following statements hold:
		\begin{subequations}\label{sts}
			\begin{align}
				\dot{\Gamma}_{\mathbb{A}}&=\Gamma_{\mathbb{A}}(\Ad^{-1}_{\mathcal{T}_{u_c}(\theta_k)}\bar{\xi}_k+\nu_k u_c)^{\wedge}\label{st_a} \\
				\bar{\psi}_{\nabla}(\bar{X}_k, \theta_k)&=\Ad^{-\top}_{\mathcal{T}_{u_c}(\theta_k)}\bar{\psi}((I_4-\Gamma_{\mathbb{A}}^{-1})\mathbb{A}) \label{st_b}\\
				\nabla_{\theta_k}U(\bar{X}_k, \theta_k)&=\gamma \theta_k+2u_c^{\top}\bar{\psi}((I_4-\Gamma_{\mathbb{A}}^{-1})\mathbb{A}) \label{st_c}\\
				\mathcal{A}_0 &\in \Upsilon_U: = \Upsilon_V \times \{0\} \label{U_cri}
			\end{align}
		\end{subequations} 
		where $\bar{\psi}_{\nabla}(\bar{X}_k, \theta_k):=\bar{\psi}(\bar{X}_k^{-1}\nabla_{\bar{X}_k}U(\bar{X}_k, \theta_k))$ for the sake of simplicity.
	\end{lem}
	
	\par Next, we introduce the following lemma, adapted from \citet[Proposition 2]{wang2021hybrid}, that will be useful for our hybrid feedback control design.
	
	\begin{lem}   \label{lem:parameters}
		Consider the potential function $U$ on $SE(3)\times \mathbb{R}$ in \eqref{U_def} with respect to $\mathcal{A}_0$. Given $\mathbb{A}$ in \eqref{AA_def}, the set $\mathcal{P}_\mathbb{A}:= \{\Theta, u_{c1}, u_{c2}, \gamma, \delta_{\bar{X}}\}$ is given by 
		\begin{align} \label{eqn:P_A}
			\mathcal{P}_\mathbb{A}: \left\{
			\begin{aligned}
				\Theta &= \{|\theta_{\jmath}| \in (0,\pi], \jmath=1,\dots,m\}  \\
				u_{c1} &= \alpha_1 v_1^{W}+\alpha_2 v_2^{W}+\alpha_3 v_3^{W} \\
				u_{c2} & =-u_{c1}^{\times}bd^{-1}\\
				\gamma &< \textstyle \frac{4\Delta_W^*}{\pi^2}\\
				\delta_{\bar{X}}& \textstyle <\big(\frac{4\Delta_W^*}{\pi^2}-\gamma\big)\frac{\theta_{M}^2}{2},~ \theta_{M}:= \sup\nolimits_{\theta_k'\in \Theta} |\theta_k'|
			\end{aligned}
			\right.
		\end{align}
		where the constant scalars $0\leq \alpha_1, \alpha_2,\alpha_3\leq 1$ and $\Delta^*_W =\min_{v \in \mathcal{E}_v(W)} \Delta(v,u_{c1}) =\min_{v \in \mathcal{E}_v(W)} u_{c1}^{\top}(\tr(W)I_3-W-2v^{\top}Wv(I_3-vv^{\top}))u_{c1}>0$ are chosen as per one of the following three cases:
		\begin{itemize}
			\item [1)] if $\lambda_1^W = \lambda_2^W$, $\alpha_3^2 = 1 - \tfrac{\lambda_2^W}{\lambda_3^W}, \alpha_i^2=\tfrac{\lambda_2^W}{2\lambda_3^W}, i\in\{1,2\}$, and $\Delta_W^* = \lambda_1^W\bigl(1 - \tfrac{\lambda_2^W}{\lambda_3^W}\bigr)$;
			\item [2)] if $\lambda_2^W \geq \tfrac{\lambda_1^W \lambda_3^W}{\lambda_3^W - \lambda_1^W}$, $\alpha_1^2=0, \alpha_i^2 = \tfrac{\lambda_i^W}{\lambda_2^W + \lambda_3^W}, i\in\{2,3\}$, and $\Delta_W^* = \lambda_1^W$;
			\item [3)] if $\lambda_1^W < \lambda_2^W < \tfrac{\lambda_1^W \lambda_3^W}{\lambda_3^W - \lambda_1^W}$, $\alpha_i^2 = 1 - \tfrac{4 \prod_{j\neq i} \lambda_j^W}{\sum_{\imath=1}^3 \sum_{k\neq \imath} \lambda_{\imath}^W \lambda_k^W}$, $i\in\{1,2,3\}~ \text{and} ~ \Delta_W^* = \tfrac{4 \prod_{j} \lambda_j^W}{\sum_{\imath=1}^3 \sum_{k\neq \imath} \lambda_{\imath}^W \lambda_k^W}$.
		\end{itemize} 
		with $(\lambda_i^W,v_i^W)$ denoting the $i$-th pair of eigenvalue and eigenvector of $W$. Then,  for every $k \in \mathcal{M}$, one has
		\begin{equation}\label{mu_U_def}
			\mu_{U}(\bar{X}_k, \theta_k):=U(\bar{X}_k, \theta_k)-\min_{\theta_k' \in \Theta}U(\bar{X}_k, \theta'_k)>\delta_{\bar{X}}
		\end{equation}
		for all $(\bar{X}_k, \theta_k) \in \Upsilon_U \setminus 
		\{\mathcal{A}_0\}$ with some constant $\delta_{\bar{X}}>0$ given by \eqref{eqn:P_A}.
	\end{lem}
	\begin{pf}
		See Appendix A.
	\end{pf}
	\begin{rem}
		According to Lemma \ref{lem:parameters}, the potential function $U$ in \eqref{U_explicit} can be constructed from the conditions of $\mathcal{P}_{\mathbb{A}}$ in \eqref{mu_U_def}. The choice of the angular warping direction $u_{c1} \in \mathbb{S}^2$ and  $\Delta_W^*$ is inspired by \citet[Proposition 2]{berkane2016construction}, which is important for robustness with respect to measurement noise. The choice of the parameters $\Theta$, $\gamma$ and $\delta_{\bar{X}}$ is inspired by \citet[Proposition 2]{wang2021hybrid}, which is the key to constructing the potential function $U$ on $SE(3)\times \mathbb{R}$. Moreover, the choice of the vector $u_{c2} \in \mathbb{R}^3$ is crucial for the design of the potential function $U$ on $SE(3)\times \mathbb{R}$ for a general weighting matrix $\mathbb{A}$. As shown in Lemma \ref{lem:parameters}, the vector $u_{c2}$ in this work is designed in terms of $b, d, u_{c1}$ according to the matrix $\mathbb{A}$.
		This enables the potential function $U$ in \eqref{U_def} can be simplified to \eqref{U_explicit} such that the inequality \eqref{mu_U_def} with a constant positive gap $\delta_{\bar{X}}$ can be easily obtained. 
	\end{rem}
	
	Let $x:=(\bar{X}_1, \theta_1, \dots, \bar{X}_M, \theta_M) \in \mathcal{S}$ denote the extended state with $\mathcal{S}:=(\SE(3) \times \mathbb{R})^M$ denoting the extended state space. Based on the potential function $U$ on $\SE(3) \times \mathbb{R}$ in \eqref{U_def}, we construct the following synergetic potential function $\bar{U}$ on $\mathcal{S}$ with respect to $\mathcal{A}:=\{x\in \mathcal{S}: (\bar{X}_k, \theta_k)=\mathcal{A}_0, \forall k \in \mathcal{M}\}$:
	\begin{align}\label{barU_def}
		\textstyle \bar{U}(x) =\sum_{k=1}^M U(\bar{X}_k, \theta_k).  
	\end{align}
	From \eqref{barU_def}, for each $k \in \mathcal{M}$, the gradients of $\bar{U}$ with respect to $\bar{X}_k$ and $\theta_k$ are given by
	\begin{subequations}\label{simlify}
		\begin{align}
			\nabla_{\bar{X}_k} \bar{U}(x)&=\nabla_{\bar{X}_k} U(\bar{X}_k, \theta_k)  \\
			\nabla_{\theta_k} \bar{U}(x)&=\nabla_{\theta_k} U(\bar{X}_k, \theta_k)
		\end{align}
	\end{subequations}
	Then, the set of all critical points of $\bar{U}$ is given by 
	$\Upsilon_{\bar{U}}:=\{x\in \mathcal{S}: \nabla_{\bar{X}_k} \bar{U}(x)=\mathbf{0}, \nabla_{\theta_k} \bar{U}(x)=0, \forall k \in \mathcal{M}\}$. 
	
	\section{Hybrid Pose Synchronization Scheme}\label{V}
	In this section, by making use of the framework of hybrid dynamical systems  \citep{goebel2009hybrid}, a hybrid feedback control scheme for global pose synchronization of the system \eqref{eqn:kinematic} is considered.
	Given the matrix $\mathbb{A}$ in \eqref{AA_def} and the set $\mathcal{P}_{\mathbb{A}}$ in \eqref{eqn:P_A}, a potential function $U$ satisfying the condition \eqref{mu_U_def} in Lemma \ref{lem:parameters} can be constructed. Then, the hybrid dynamics of the switching variable $\theta_k$ that allow continuous flows and discrete jumps will be designed to avoid the undesired critical points, leaving $\mathcal{A}_0$ as the unique attractor. Inspired by \cite{wang2021hybrid}, for each $k\in \mathcal{M}$, we propose the following hybrid dynamics:
	\begin{align}\label{mathcalHtheta}
		\mathcal{H}_{\theta_k}:
		\left\{ \begin{aligned}
			&\dot{\theta}_k=-k_{\theta}\nabla_{\theta_k}U(\bar{X}_k,\theta_k), \quad ~(\bar{X}_k, \theta_k) \in \mathcal{F}_k\\
			&\theta_k^+ \in g(\bar{X}_k, \theta_k), \quad \quad \quad \quad \quad (\bar{X}_k, \theta_k) \in \mathcal{J}_k 
		\end{aligned}\right.
	\end{align}
	with $k_{\theta}>0$ and the flow and jump sets given by
	\begin{align}\label{set_Fk_J_k}
		&\mathcal{F}_k:=\{(\bar{X}_k, \theta_k) \in \SE(3) \times \mathbb{R}: \mu_{U}(\bar{X}_k, \theta_k) \leq \delta_{\bar{X}} \} \nonumber \\ 
		&\mathcal{J}_k:= \{(\bar{X}_k, \theta_k) \in \SE(3) \times \mathbb{R}: \mu_{U}(\bar{X}_k, \theta_k) \geq \delta_{\bar{X}}\}
	\end{align}
	where $\mu_{U}(\bar{X}_k, \theta_k)$ and $\delta_{\bar{X}}$ are given by Lemma ~\ref{lem:parameters}, and the jump map  $g: \SE(3) \times \mathbb{R} \rightrightarrows \Theta$ is given by
	\begin{align}\label{eqn: gtheta}
		g(\bar{X}_k, \theta_k):=\left\{\theta_k \in \Theta: \theta_k= \arg\min_{\theta_k' \in \Theta}U(\bar{X}_k, \theta'_k) \right\}. 
	\end{align}
	Note that the switching variable $\theta_k$ flows when the state $(\bar{X}_k, \theta_k)$ is away from the undesired critical set $\Upsilon_U \setminus \{\mathcal{A}_0\}$, and jumps to some $\theta_k \in \Theta$, leading to minimum value of $U(\bar{X}_k,\theta_k')$, when the state $(\bar{X}_k, \theta_k)$ is in the neighborhood of the set $\Upsilon_U\setminus \{\mathcal{A}_0 \}$.
	From the definitions of $\mathcal{J}_k$ in \eqref{set_Fk_J_k}, $g$ in \eqref{eqn: gtheta} and $\mu_{U}$ in \eqref{mu_U_def}, one can verify that the inequality $\mu_{U}(\bar{X}_k, \theta_k) \geq \delta_{\bar{X}}$ holds for all $(\bar{X}_k, \theta_k) \in \mathcal{J}_k$, which can guarantee a minimum decrease of the potential function $U$ by a constant gap $\delta_{\bar{X}}$ after each jump. Then, we propose the following distributed hybrid feedback control law:
	\begin{align}\label{hybridcontroller}
		\begin{aligned}
			&\underbrace{
				\begin{aligned}
					u_{i} &= -k_X \sum_{\kappa =1}^M\bar{B}_{i\kappa} \bar{\psi}_{\nabla}(\bar{X}_{\kappa}, \theta_{\kappa})
					-k_{\xi}\xi_i -k_e\sum_{j \in \mathcal{N}_i}(\xi_i-\xi_j) \\[2pt]
					\dot{\theta}_k &= -k_{\theta}\nabla_{\theta_k}U(\bar{X}_k,\theta_k),  
				\end{aligned}
			}_{x \in \mathcal{F}_i:=\{x\in \mathcal{S}:~ \forall \kappa \in \mathcal{M}_i^+,~ (\bar{X}_\kappa, \theta_\kappa) \in \mathcal{F}_\kappa\}} \\[2pt]
			&\underbrace{
				\begin{aligned}
					&\hphantom{u_{i} = -k_X \sum_{k =1}^M\bar{B}_{ik} \bar{\psi}_{\nabla}(\bar{X}_k, \theta_k) 
						-k_{\xi}\xi_i -k_e\sum_{j \in \mathcal{N}_i}(\xi_i-\xi_j) }\\[-8pt]
					&\theta^+_k \in g(\bar{X}_k, \theta_k), 
				\end{aligned}
			}_{x \in \mathcal{J}_i:=\{x\in \mathcal{S}:~ \exists \kappa \in \mathcal{M}_i^+,~ (\bar{X}_\kappa, \theta_\kappa) \in \mathcal{J}_\kappa\}}
		\end{aligned}
	\end{align}
	for all $i \in \mathcal{V}$ and $k \in \mathcal{M}_i^+$, with $ k_X,k_{\xi},k_e>0$, $\mathcal{F}_k$ and $\mathcal{J}_k$ in \eqref{set_Fk_J_k}, $g$ in \eqref{eqn: gtheta} and the incidence matrix $\bar{B}$ in \eqref{Bt_def}.

	\begin{rem}
		The distributed hybrid feedback control law $u_i$ in \eqref{hybridcontroller}  consists of three terms: the first gradient-based relative pose feedback term 
		ensures the pose synchronization for MRBSs; the second velocity feedback term provides a damping that ensures the group velocity $\xi_i$ converging to zero; 
		the last relative velocity feedback term provides relative velocity information for MRBSs that can improve the transient performance.     
	\end{rem}
	
	Define the new states $\bar{x}:=(x, \xi) \in \bar{\mathcal{S}}:=(\SE(3) \times \mathbb{R})^M \times \mathbb{R}^{6N}$ and $\xi:=(\xi_1, \dots, \xi_N)\in \mathbb{R}^{6N}$. From \eqref{eqn:errdynamics} and \eqref{hybridcontroller}, one obtains the following hybrid closed-loop system:
	\begin{align}\label{hybsys}
		\mathcal{H}: \left\{
		\begin{array}{ll} 
			\dot{\bar{x}}~~=\bar{F}(\bar{x}),\quad &\bar{x} \in \bar{\mathcal{F}}:=\{\bar{x} \in \bar{\mathcal{S}}: x \in \mathcal{F} \}\\
			\bar{x}^+ \in \bar{G}(\bar{x}), \quad &\bar{x} \in \bar{\mathcal{J}}:=\{\bar{x} \in \bar{\mathcal{S}}: x \in \mathcal{J} \}
		\end{array} \right.
	\end{align}
	where $\mathcal{F}=\bigcap_{i=1}^N \mathcal{F}_i, \mathcal{J}=\bigcup_{i=1}^N \mathcal{J}_i$ and 
	\begin{align*}\small 
		\bar{F}(\bar{x}):=\begin{pmatrix}
			\bar{X}_1 \bar{\xi}_1^{\wedge} \\
			-k_{\theta}\nabla_{\theta_1}U(\bar{X}_1, \theta_1)\\
			\vdots\\
			\bar{X}_M \bar{\xi}_M^{\wedge} \\
			-k_{\theta}\nabla_{\theta_M}U(\bar{X}_M, \theta_M)\\
			\mathbb{I}_1^{-1}(\ad^{\top}_{\xi_1}\mathbb{I}_1\xi_1+u_{1})\\
			\vdots\\
			\mathbb{I}_N^{-1}(\ad^{\top}_{\xi_N}\mathbb{I}_N\xi_N+u_{N})\\
		\end{pmatrix},
		\bar{G}(\bar{x}) :=	\begin{pmatrix}\bar{X}_1\\
			g(\bar{X}_1, \theta_1)\\
			\vdots\\
			\bar{X}_M\\
			g(\bar{X}_M, \theta_M)\\
			\xi_1\\
			\vdots\\
			\xi_N \end{pmatrix}.
	\end{align*}
	with $g$ defined in \eqref{eqn: gtheta} and $u_i$  defined in \eqref{hybridcontroller}. One can verify that $\bar{\mathcal{F}}\cup \bar{\mathcal{J}}= \bar{\mathcal{S}}$, $\bar{\mathcal{F}}$ and $\bar{\mathcal{J}}$ are closed, and the hybrid  system \eqref{hybsys} is autonomous and satisfies the hybrid basic conditions \citep{goebel2009hybrid}.
	Now, one can state the following main result:
	\begin{thm}\label{T1}
		Consider the hybrid closed-loop system \eqref{hybsys} with $g$ in \eqref{eqn: gtheta} and $u_i$ in \eqref{hybridcontroller}. Let $k_X, k_{\xi}, k_e>0$, and suppose that Assumption \ref{assum:1} holds and the set $\mathcal{P}_\mathbb{A}$ is chosen as per Lemma \ref{lem:parameters}. Then, the closed set $\bar{\mathcal{A}}:=\{\bar{x} \in \bar{\mathcal{S}}:x\in \mathcal{A}, \xi=\mathbf{0} \}$ is globally asymptotically stable for the hybrid closed-loop system \eqref{hybsys} and the number of jumps is finite.
	\end{thm}
	\begin{pf}
		Consider the following Lyapunov function candidate:
		\begin{align}\label{V_def}
			\textstyle \bar{V}(\bar{x})=k_X\bar{U}(x)+\sum_{i=1}^N \xi_i^{\top} \mathbb{I}_i \xi_i.
		\end{align}
		From \eqref{eqn:errdynamics},\eqref{U_def},\eqref{barU_def},\eqref{mathcalHtheta} and Lemma ~\ref{lem:Uproperty}, the time-derivative of $U(\bar{X}_k, \theta_k)$ is given by
		\begin{align}\label{U_derivative}
			&\dot{U}(\bar{X}_k, \theta_k) \nonumber \\
			&\qquad =\langle \nabla_{\bar{X}_k}U(\bar{X}_k, \theta_k), \bar{X}_k\bar{\xi}_k^{\wedge} \rangle_{\bar{X}_k} +\langle\langle \nabla_{\theta_k}U(\bar{X}_k, \theta_k), \dot{\theta}_k \rangle\rangle  \nonumber \\
			&\qquad=\langle\langle \bar{X}_k^{-1}\nabla_{\bar{X}_k}U(\bar{X}_k, \theta_k), \bar{\xi}_k^{\wedge} \rangle\rangle+\nabla_{\theta_k}U(\bar{X}_k, \theta_k)\dot{\theta}_k  \nonumber  \\
			&\qquad=2\bar{\xi}_k^{\top} \bar{\psi}_{\nabla}(\bar{X}_k, \theta_k)-k_{\theta}|\nabla_{\theta_k}U(\bar{X}_k, \theta_k)|^2
		\end{align}
		it follows from \eqref{barU_def} that
		\begin{align}\label{eqn:dot_bar_U}
			\dot{\bar{U}}(x)
			&\textstyle =\sum_{k=1}^M (2\bar{\xi}_k^{\top} \bar{\psi}_{\nabla}(\bar{X}_k, \theta_k)-k_{\theta}|\nabla_{\theta_k}U(\bar{X}_k, \theta_k)|^2) \nonumber \\
			&=2\bar{\xi}^{\top}\Psi^{\bar{X}}_{\nabla}(x)-k_{\theta}\|\Psi_{\nabla}^{\theta}(x)\|^2 \nonumber \\ 
			&=2\xi^{\top}\bar{B}(t)\Psi^{\bar{X}}_{\nabla}(x)-k_{\theta}\|\Psi_{\nabla}^{\theta}(x)\|^2
		\end{align}
		where $\Psi^{\bar{X}}_{\nabla}(x):=[\bar{\psi}_{\nabla}(\bar{X}_1, \theta_1)^{\top},\dots,
		\bar{\psi}_{\nabla}(\bar{X}_M, \theta_M)^{\top}]^{\top} 
		\in \mathbb{R}^{6M}$ and $\Psi_{\nabla}^{\theta}(x):=[\nabla_{\theta_1}U(\bar{X}_1, \theta_1),\dots,
		\nabla_{\theta_M}U(\bar{X}_M, \theta_M)]^{\top}\\ \in \mathbb{R}^M$ and  we made use of the facts \eqref{identa}, \eqref{eqn:bar_xi} and \eqref{simlify}. 
		
		From \eqref{eqn:errdynamics} and \eqref{hybridcontroller}, one can show that
		\begin{align}\label{eqn:sum_dot_xi_i}
			\textstyle \sum_{i=1}^N \xi_i^{\top}\mathbb{I}_i \dot{\xi}_i
			&\textstyle =-k_X\sum_{i=1}^N \xi_i^{\top} \sum_{\kappa =1}^M\bar{B}_{i\kappa} \bar{\psi}_{\nabla}(\bar{X}_{\kappa}, \theta_{\kappa})\nonumber \\  
			&\textstyle ~~~~ -\sum_{i=1}^N \xi_i^{\top} \big(k_{\xi}\xi_i + k_e \sum_{j \in \mathcal{N}_i}(\xi_i-\xi_j) \big)  \nonumber \\
			&\textstyle = -k_X\xi^{\top}\bar{B}(t)\Psi^{\bar{X}}_{\nabla}(x)-k_{\xi}\| \xi \|^2-k_e\|\mathcal{B} \xi \|^2
		\end{align}
		where $\mathcal{B}:= B^{\top} \otimes I_6$ with $\otimes$ denoting the Kronecker product, and we made use of the facts: $\ad_{\xi_i}\xi_i=\mathbf{0}$ in \eqref{propc} and $\sum_{i=1}^N \xi_i^{\top}\sum_{j \in \mathcal{N}_i}(\xi_i-\xi_j)=\xi^{\top}(BB^{\top}\otimes I_6)\xi  =\|\mathcal{B} \xi \|^2$ with $BB^{\top} \in \mathbb{R}^{N\times N}$ known as the Laplacian matrix.
		Then, by virtue of \eqref{eqn:dot_bar_U} and \eqref{eqn:sum_dot_xi_i}, one obtains
		\begin{align}\label{eqn:re_dot_V}
			\dot{\bar{V}}(\bar{x}) 
			& \textstyle = k_X\dot{\bar{U}}(x)+2\sum_{i=1}^N \xi_i^{\top} \mathbb{I}_i \dot{\xi}_i \nonumber \\
			& = -k_Xk_{\theta}\|\Psi_{\nabla}^{\theta}(x)\|^2-2k_{\xi}\| \xi \|^2-2k_e\|\mathcal{B}  \xi \|^2  
		\end{align}
		for all $\bar{x}\in \bar{\mathcal{F}}$. This implies that $\dot{\bar{V}}(\bar{x}) \leq 0$ and $\bar{V}(\bar{x})$ is non-increasing along the flow of  \eqref{hybsys}.

		On the other hand, for each jump $\bar{x} \in \bar{\mathcal{J}}$, from the definition of $\bar{\mathcal{J}}$ there exists a nonempty set $\mathcal{M}' \subseteq \mathcal{M}$ and at least one edge $k \in \mathcal{M}$ such that $\mu_U(\bar{X}_k, \theta_k) \geq \delta_{\bar{X}}$ with $\mu_U$ given by \eqref{mu_U_def}. Together with \eqref{barU_def}, one has 
		\begin{align}\label{ineq:V_diff}
			&\bar{V}(\bar{x}^+)-\bar{V}(\bar{x})=k_X\bar{U}(x^+)-k_X\bar{U}(x)  \nonumber\\
			&\textstyle \qquad \quad  =-k_X \sum_{k\in \mathcal{M}'}(U(\bar{X}_k, \theta_k)-\min_{\theta_k' \in \Theta}U(\bar{X}_k, \theta'_k)) \nonumber \\ 
			&\textstyle \qquad \quad =-k_X\sum_{k\in \mathcal{M}'} \mu_U(\bar{X}_k, \theta_k)  \nonumber \\ 
			&\textstyle \qquad \quad \leq -k_X \sum_{k\in \mathcal{M}'} \delta_{\bar{X}}\leq -k_X \delta_{\bar{X}}
		\end{align}
		for all $\bar{x}\in \bar{\mathcal{J}}$. This implies that $\bar{V}(\bar{x})$ is strictly decreasing
		over the jumps of \eqref{hybsys}. In view of \eqref{eqn:re_dot_V} and \eqref{ineq:V_diff}, it follows from \citet[Theorem 23]{goebel2009hybrid} that $\bar{\mathcal{A}}$ is stable, and hence every maximal solution to the hybrid closed-loop system \eqref{hybsys} is bounded. Inspired by the Proof of \citet[Theorem 1]{wang2021hybrid}, one can obtain from \eqref{eqn:re_dot_V} and \eqref{ineq:V_diff} that $\bar{V}(\bar{x}(t,j)) \leq \bar{V}(\bar{x}(t_j,j))  \leq \bar{V}(\bar{x}(t_j,j-1))-k_X\delta_{\bar{X}}$ for all $(t,j),(t_j,j),(t_j,j-1) \in \dom \bar{x}$ with $(t,j) \succeq (t_j,j) \succeq (t_j,j-1)$. 
		Therefore, it is clear that $0 \leq \bar{V}(\bar{x}(t,j)) \leq \bar{V}(\bar{x}(0,0))-j k_X\delta_{\bar{X}}$ for all $(t,j)\in \dom \bar{x}$, which leads to  $j \leq \lceil \bar{V}(\bar{x}(0,0))/k_X\delta_{\bar{X}} \rceil$, where $\lceil \cdot \rceil$ denotes the ceiling function. 
		Thus, one concludes that the number of jumps is finite and depends on the initial conditions. 
		
		Furthermore, applying the invariance principle for hybrid systems \citep{goebel2009hybrid}, it follows that any solution to the hybrid closed-loop system \eqref{hybsys} converges to the largest invariant set contained in $\varGamma:= \{\bar{x} \in \bar{\mathcal{F}}: \Psi_{\nabla}^{\theta}(x)=\mathbf{0}, \xi=\mathbf{0}\}$. For each $\bar{x}\in \varGamma$, from $\xi \equiv \mathbf{0}$, it follows that $\dot{\xi}=\mathbf{0}$. This implies $\sum_{\kappa = 1}^M\bar{B}_{i\kappa} \bar{\psi}_{\nabla}(\bar{X}_{\kappa}, \theta_{\kappa})=\mathbf{0}$ (\ie, $\bar{B}(t)\Psi_{\nabla}^{\bar{X}}(x)=\mathbf{0}$) from \eqref{eqn:errdynamics} and \eqref{hybridcontroller}.
		Then, by virtue of Lemma \ref{lem: HE}, it follows that $\Psi_{\nabla}^{\bar{X}}(x)=\mathbf{0}$. From $\Psi_{\nabla}^{\bar{X}}(x)=\mathbf{0}$ and $\Psi_{\nabla}^{\theta}(x)=\mathbf{0}$, one has $x \in \Upsilon_{\bar{U}}$. Consequently, any solution to the hybrid closed-loop system \eqref{hybsys} converges to the largest weakly invariant set contained in $\varGamma':=\{\bar{x} \in \bar{\mathcal{S}}: x \in \mathcal{F} \cap \Upsilon_{\bar{U}}, \xi=\mathbf{0}\}$.
		On the other hand, given $x \in \mathcal{A}$, one has, for all $k \in \mathcal{M}, \mu_U(\bar{X}_k, \theta_k)=-\min_{\theta_k' \in \Theta}U(\bar{X}_k, \theta_k') \leq 0$, which implies that $\mathcal{A} \subset \mathcal{F} \cap \Upsilon_{\bar{U}}$ and $\mathcal{F} \cap (\Upsilon_{\bar{U}} \setminus \mathcal{A})=  \varnothing$. In addition, applying some set-theoretic arguments, one has $\mathcal{F} \cap \Upsilon_{\bar{U}} \subset (\mathcal{F} \cap (\Upsilon_{\bar{U}} \setminus \mathcal{A})) \cup (\mathcal{F} \cap \mathcal{A})= \varnothing \cup \mathcal{A}$. From $\mathcal{A} \subset \mathcal{F} \cap \Upsilon_{\bar{U}}$ and $\mathcal{F} \cap \Upsilon_{\bar{U}} \subset \mathcal{A}$, one has $\mathcal{F} \cap \Upsilon_{\bar{U}}=\mathcal{A}$. Thus, one can verify that $\varGamma'= \bar{\mathcal {A}}$. Note that every maximal solution to the hybrid system \eqref{hybsys} is bounded,  $\bar{F}(\bar{x}) \subset T_{\bar{\mathcal{F}}}(\bar{x})$ for any $\bar{x} \in \bar{\mathcal{F}} \setminus \bar{\mathcal{J}}$ with $T_{\bar{\mathcal{F}}}(\bar{x})$ denoting the tangent cone to $\bar{\mathcal{F}}$ at the point $\bar{x}$, and $\bar{G}(\bar{x}) \subset\bar{\mathcal{F}} \cup \bar{\mathcal{J}}=\bar{\mathcal{S}}$.  Therefore, in terms of \citet[Theorem S3]{goebel2009hybrid}, it is clear that every maximal solution to \eqref{hybsys} is complete.
		Finally, together with the fact that the hybrid closed-loop system \eqref{hybsys} satisfies the basic hybrid conditions \citep{goebel2009hybrid}, one can conclude the set $\bar{\mathcal{A}}$ is globally asymptotically stable. This completes the proof.
	\end{pf}
	
	From the potential function $\bar{U}$ given in \eqref{barU_def} with $U$ constructed as \eqref{U_def},  together with \eqref{st_b}, \eqref{st_c} and \eqref{eqn: gtheta}, the explicit form of the distributed hybrid feedback law \eqref{hybridcontroller} can be derived as follows: 
	\begin{align}\label{ex_controller}
		\begin{aligned}
			&\underbrace{
				\begin{aligned}
					u_{i}&\textstyle  =-k_X  
					\sum\nolimits_{\kappa \in \mathcal{M}_i} \bar{B}_{i\kappa} \Ad^{-\top}_{\mathcal{T}_{u_c}(\theta_\kappa)}\bar{\psi}((I_4-\mathcal{T}_{u_c}^{-1}(\theta_\kappa)\bar{X}_\kappa^{-1})\mathbb{A})\\
					&\textstyle  \quad -k_{\xi}\xi_i -k_e\sum\nolimits_{j \in \mathcal{N}_i}(\xi_i-\xi_j) \\  
					\dot{\theta}_k&=-k_{\theta}\big(
					\gamma \theta_k+2u^{\top}\bar{\psi}((I_4-(\bar{X}_k \mathcal{T}_{u_c}(\theta_k))^{-1})\mathbb{A})
					\big) 
				\end{aligned} 
			}_{x\in \mathcal{F}_i }\\  
			&\underbrace{
				\begin{aligned}
					&\hphantom{ u_{i}\textstyle  =-k_X  
						\sum\nolimits_{\kappa \in \mathcal{M}_i} \bar{B}_{i\kappa} \Ad^{-\top}_{\mathcal{T}_{u_c}(\theta_\kappa)}\bar{\psi}((I_4-\mathcal{T}_{u_c}^{-1}(\theta_\kappa)\bar{X}_\kappa^{-1})\mathbb{A})}\\[-8pt]
					&\theta^+_k  \in \{\theta_k \in \Theta: \theta_k= \arg \min\nolimits_{\theta_k' \in \Theta}U(\bar{X}_k, \theta'_k)\}  
				\end{aligned}
			}_{x\in \mathcal{J}_i  } 
		\end{aligned}
	\end{align}
	for all $i\in \mathcal{V}, k \in \mathcal{M}_i^+ $, where we made use of the fact from \eqref{Bt_def} that $\bar{B}_{i\kappa}=\mathbf{0}$ for all $\kappa \notin \mathcal{M}_i=\mathcal{M}_i^+ \cup \mathcal{M}_i^-$.
	
	\section{SIMULATION}
	In this section, numerical simulation results are provided to illustrate the performance of the proposed distributed hybrid feedback controller \eqref{ex_controller}. We consider a set of 6 fully actuated UAVs with the undirected communication topology shown in Fig. \ref{fig:graph} (a), and the neighbor sets are given as $\mathcal{N}_1=\{2\}, \mathcal{N}_2=\{1,3,5\}, \mathcal{N}_3=\{2,4\}, \mathcal{N}_4=\{3\}, \mathcal{N}_5=\{2,6\}$ and $\mathcal{N}_6=\{5\}$. We assign an arbitrary orientation to the graph $\mathcal{G}$ as shown in Fig.  \ref{fig:graph} (b). For each UAV system $i\in \mathcal{V}$, the mass and inertia matrix are taken as $m_i=2.4 \kg$ and $J_i=$diag$ ([0.043,0.041,0.082]) \kg \cdot \m^2$.
	
	\begin{figure}[!ht]
		\centering
		\includegraphics[width=0.42\textwidth]{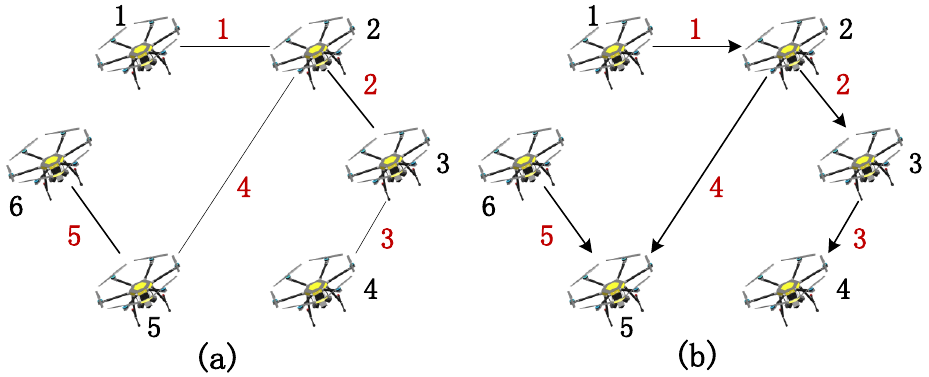}
		\vspace{-0.16in}
		\caption{The topology graph $\mathcal{G}$: (a) undirected; (b) undirected graph with an arbitrary orientation.} 
		\label{fig:graph} 
	\end{figure}
	
	Consider the matrix $\mathbb{A}$ in \eqref{AA_def} for the proposed potential function $U$ in  \eqref{U_def} as
	$$A=\begin{bmatrix}
		69.02&   56.08  & 61.19  \\
		56.08 &  56.25  & 51.17 \\
		61.19 &  51.17  & 57.66  
	\end{bmatrix}, b= \begin{bmatrix}
		18.24\\ 
		14.43\\
		15.96
	\end{bmatrix}, d=5$$ 
	and choose the set $\mathcal{P}_{\mathbb{A}}$ in \eqref{eqn:P_A} as  
	$\Theta=\{ 0.3 \pi\}, u_{c1}=[0.11, 0.99, 0.04]^{\top},  u_{c2} =-u_{c1}^{\times}bd^{-1}=[-3.06,0.20,3.32]^{\top}, \\ \gamma=0.33, \delta_{\bar{X}}=0.02$ and $\Delta_W^*=\lambda_1^{W}=0.9$. In addition, we consider the following initial conditions: 
	$\xi(0)=0, \theta_k(0)=0$ for $k\in \mathcal{M}$, $R_1(0)=\mathcal{R}_a(-\frac{\pi}{2}, v), R_2(0)=\mathcal{R}_a(\frac{\pi}{2}, v), R_3(0)=\mathcal{R}_a(-\frac{\pi}{2}, v)$,
	$R_4(0)=\mathcal{R}_a(\frac{\pi}{2}, v), R_5(0)=\mathcal{R}_a(-\frac{\pi}{2}, v), R_6(0)=\mathcal{R}_a(\frac{\pi}{2}, v)$, $v=[ 0.2686, 0.8549, 0.4438]^{\top}$, $p_i(0)=(I_3-R_i(0))bd^{-1}$ for $i \in \mathcal{V}$. Note that the initial states are specifically chosen to ensure that the relative pose of each edge is close to one of the undesired equilibria.
	The gains of the distributed hybrid feedback controller \eqref{ex_controller} are set to $k_X = 100, k_{\xi} = 1, k_e = 0.6, k_{\theta} = 1$. From Fig.~\ref{fig:simulation results}, one can see that the states $\bar{R}_k, \bar{p}_k, \omega_i$ and $ v_i$ for every $k \in \mathcal{M}$ and $i\in \mathcal{V}$ converge to zero as $t\rightarrow \infty$. Moreover, the variable $\theta_k$ for every $k \in \mathcal{M}$ jumps from $0$ to $0.3 \pi$ at $t=0$, and then converges to zero as $t\rightarrow \infty$.
	
	\begin{figure}[!ht]
		\centering
		\includegraphics[width=0.48\textwidth]{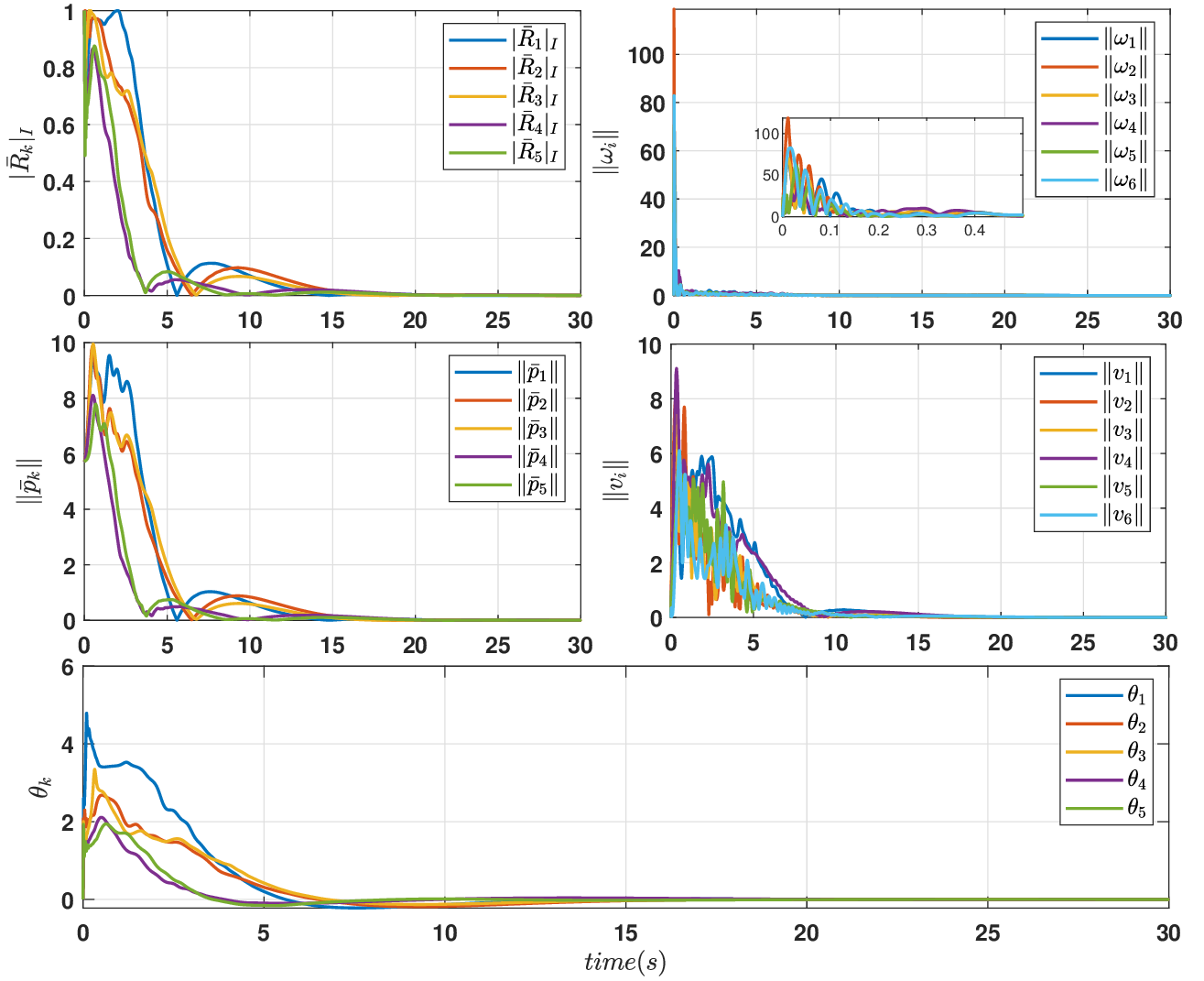}
		\label{simulation results} \\[-0.01in]
		\vspace{-0.1in}
		\caption{Simulation results for the distributed hybrid feedback \eqref{ex_controller}.} 
		\label{fig:simulation results}
	\end{figure}

	\section{Conclusion}
	We proposed a new distributed hybrid pose synchronization scheme, with GAS guarantees, on $\SE(3)$ over undirected, connected, and acyclic graphs. A new potential function on $(\SE(3) \times \mathbb{R})^{M}$ for a general weighting matrix $\mathbb{A}$, involving a set of variables $\theta_k$ with flow and jump dynamics, was constructed to help with the design of the proposed hybrid control scheme. The proposed feedback control scheme, designed on Lie group $\SE(3)$, is restricted to fully actuated rigid body systems. Since many UAVs are underactuated (for instance, quadrotor UAVs), the pose control scheme proposed in this paper cannot be applied directly. Designing global pose synchronization schemes on  $\SE(3)$, for underactuated rigid body systems, over general communication graph topologies, including directed graphs, is an interesting future work.

	
	\appendix
	\section{Proof of Lemma \ref{lem:parameters}}  
	The proof of this lemma is inspired by \citet[Proposition 2]{wang2021hybrid}. 
	From  \eqref{def:Upsilon_V} and \eqref{U_cri}, the set of undesired critical points is given by $ \Upsilon_U \setminus \{\mathcal{A}_0\} = \{(\bar{X}_k,\theta_k) \in SE(3) \times \mathbb{R}: \bar{X}_k=\mathcal{T}(\bar{R}_k, \bar{p}_k), \bar{R}_k=\mathcal{R}_a(\pi, v), \bar{p}_k=(I_3-R_a(\pi, v))bd^{-1}, v\in \mathcal{E}_v(W), \theta_k=0\}$. This implies that, for any $(\bar{X}_k, \theta_k) \in \Upsilon_U \setminus \{\mathcal{A}_0\}$, one has $\theta_k = 0$, $\bar{X}_k=\mathcal{T}(\bar{R}_k, \bar{p}_k)$ with $ \bar{R}_k=\mathcal{R}_a(\pi, v), \bar{p}_k =(I_3-R_a(\pi, v))bd^{-1})$ and $v \in \mathcal{E}_v(W)$. Then, from \eqref{U_explicit}, for any $(\bar{X}_k, \theta_k) \in \Upsilon_U \setminus \{\mathcal{A}_0\}$ with $\bar{X}_k=\mathcal{T}(\bar{R}_k, \bar{p}_k)$ and $\theta_k' \in \Theta \subset \mathbb{R}$, one obtains
	\begin{align}\label{U(pi,v,),0}
		U(\bar{X}_k, \theta_k)
		& \textstyle =\tr(W(I_3-\bar{R}_k\mathcal{R}_a(\theta_k, u_{c1})))+\frac{\gamma}{2}\theta_k^2 \nonumber\\
		& \textstyle ~~~~+\frac{d}{2}\|\bar{p}_k - (I_3 - \bar{R}_k)bd^{-1}\|^2 \nonumber \\
		&=\tr(W(I_3-\mathcal{R}_a(\pi,v))) 
	\end{align}
	and
	\begin{align}\label{U(theta,v,),theta}
		\small
		U(\bar{X}_k,\theta_k') 
		& \textstyle =\tr(W(I_3-\bar{R}_k\mathcal{R}_a(\theta_k', u_{c1})))+\frac{\gamma}{2}\theta_k'^2 \nonumber\\
		& \textstyle ~~~~+\frac{d}{2}\|\bar{p}_k - (I_3 - \bar{R}_k)bd^{-1}\|^2 \nonumber \\
		& \textstyle =\tr(W(I_3-\mathcal{R}_a(\pi,v)\mathcal{R}_a(\theta_k',u_{c1}))) +\frac{\gamma}{2}\theta_k'^2  \nonumber\\
		&=\tr(W(I_3-\mathcal{R}_a(\pi,v))) \nonumber \\
		& \textstyle ~~~~ +\tr(W\mathcal{R}_a(\pi,v)(I_3-\mathcal{R}_a(\theta_k',u_{c1})))+\frac{\gamma}{2}\theta_k'^2  \nonumber\\
		& \textstyle =U(\bar{X}_k,\theta_k)-2\sin^2(\frac{\theta_k'}{2})\Delta(v,u_{c1}) +\frac{\gamma}{2}\theta_k'^2  
	\end{align}
	where we have made use of the facts: $\mathcal{R}_a(\theta_k',u_{c1})=I_3+\sin (\theta_k') u_{c1}^{\times}+(1-\cos(\theta_k'))(u_{c1}^{\times})^2$ and
	$\tr(W\mathcal{R}_a(\pi,v)(I_3-\mathcal{R}_a(\theta_k',u_{c1}))
	=-2\sin^2(\frac{\theta_k'}{2})\Delta(v,u_{c1})$ with $\Delta(v,u_{c1})=u_{c1}^{\top}(\tr(W)I_3-W-2v^{\top}Wv(I_3-vv^{\top}))u_{c1}$ for all $\theta_k' \in \mathbb{R}, v\in \mathcal{E}_v(W), u_{c1}\in \mathbb{S}^2$. Given the choice of $u_{c1} \in \mathbb{S}^2$ as per the set $\mathcal{P}_{\mathbb{A}}$ in \eqref{eqn:P_A}, it  follows from \citet[Proposition 2]{berkane2016construction} that $\Delta^*=\min_{v\in \mathcal{E}_v(W)}\Delta(v,u_{c1})>0$.

	Therefore, from \eqref{U(pi,v,),0} and \eqref{U(theta,v,),theta}, for any $(\bar{X}_k, \theta_k) \in \Upsilon_U \setminus \{\mathcal{A}_0\}$ and $\Theta = \{|\theta_{\jmath}| \in (0,\pi], \jmath=1,\dots,m\} \subset \mathbb{R}$, the following inequality holds
	\begin{align}
		\mu_U(\bar{X}_k, \theta_k)& \textstyle =U(\bar{X}_k,\theta_k)-\min_{\theta_k' \in \Theta} U(\bar{X}_k,\theta_k')\nonumber\\
		& \textstyle  =\max_{\theta_k' \in \Theta}\big(2\sin^2(\frac{\theta_k'}{2})\Delta(v,u_{c1})-\frac{\gamma}{2}\theta_k'^{2}\big)\nonumber\\
		& \textstyle \geq \max_{\theta_k' \in \Theta}\big(2\sin^2(\frac{\theta_k'}{2})\Delta^*-\frac{\gamma}{2}\theta_k'^{2}\big)\nonumber\\
		& \textstyle \geq 2\sin^2(\frac{\theta_M}{2})\Delta^*-\frac{\gamma}{2}\theta_M^2 \nonumber \\
		& \textstyle \geq (\frac{4\Delta^*}{\pi^2}-\gamma)\frac{\theta_M^2}{2}>\delta_{\bar{X}}
	\end{align}
	where we made use of the facts $\frac{4\Delta^*}{\pi^2}-\gamma>0$, $\delta_{\bar{X}}<(\frac{4\Delta^*}{\pi^2}-\gamma)\frac{\theta_M^2}{2}$ and $\theta_M=\sup_{\theta_k' \in \Theta}|\theta_k'| \leq \pi$ given by \eqref{eqn:P_A}, and $2\sin^2(\frac{\theta_k'}{2})\Delta^*-\frac{\gamma}{2}\theta_k'^{2}> 0$, $|\sin(\frac{\theta_k}{2})| \geq \frac{|\theta_k|}{\pi}$ for all $|\theta_k'|\in (0,\pi]$ (\ie, $\theta_k' \in \Theta$). This completes the proof.
	
	\bibliography{ifacconf}
	
\end{document}